\documentclass[a4paper,10pt]{article}
\usepackage{a4wide}
\usepackage{fancyhdr}
\usepackage{graphicx}
\usepackage[bookmarks]{hyperref}
\usepackage{setspace}
\usepackage{microtype}
\usepackage{booktabs}   
\usepackage[utf8]{inputenc}
\usepackage{mathpartir}
\usepackage{setspace}
\usepackage{proof}
\usepackage{turnstile}
\usepackage{float}
\usepackage{mathtools}
\usepackage{stmaryrd}
\usepackage{mdframed}
\usepackage{xspace}       
\usepackage{default}
\usepackage[arrow, matrix, curve]{xy}
\usepackage{amssymb}
\usepackage{amsthm}
\usepackage{amsfonts}
\usepackage{pgf}
\usepackage{tabularx}
\usepackage{tikz}
\usepackage{qtree}
\usepackage{tikz-qtree}
\usepackage{proof}
\usepackage{turnstile}
\usepackage{float}
\usepackage{mathtools}
\usepackage[ngerman,english]{babel}
\usepackage{listings}
\usepackage{parskip}
\usepackage{default}
\usetikzlibrary{arrows,automata}
\newtheorem{observation}{Observation}
\newtheorem{example}{Example}
\newtheorem{lemma}{Lemma}

\newtheorem{claim}{Claim}
\newtheorem{definition}{Definition}

\newtheorem{theorem}{Theorem}

\title{Linear Tree Constraints}

\author{Sabine Bauer and Martin Hofmann, University of Munich}

\begin{document}
\maketitle

\begin{abstract}
 Linear tree constraints were introduced by Hofmann and Rodriguez in the context of amortized resource analysis for object oriented programs.
 More precisely, they gave a reduction from inference of resource types to constraint solving. 
 Thus, once we have found an algorithm to solve the constraints generated from a program,
 we can read off the resource consumption from their solutions.
 
 These constraints have the form of pointwise linear inequalities between infinite trees labeled with nonnegative rational numbers.
 We are interested in the question if a system of such constraints is simultaneously satisfiable.
 Bauer and Hofmann have recently identified a fragment of the tree constraint problem (UTC) that is still sufficient for program analysis and
 they proved that the list case of UTC is decidable (which was presented at LPAR-21),
 whereas the case with trees of degree at least two remained open.
 In this paper, we solve this problem.
 We give a decision procedure that covers the entire range of constraints needed for resource analysis. 
\end{abstract}

\section{Introduction and Related Work}
We start with a short overview of related work in amortized resource analysis, because that is where the constraint problem originates from.
The idea of amortized analysis goes back to the 1980s \cite{doi:10.1137/0606031}. It is an approach that takes into account not only the worst case resource consumption of programs (which may be much more than
one has in practice) but the worst case \emph{average} resource usage of \emph{sequences of operations} (cf. \cite{JH2011}). The benefit from that is that in the change of data structures during 
a computation additional resources may become available (i.e. there may be operations that bring the data into a state such that the following operations can be carried out more efficiently).
Then one knows that in the next step a better bound than the worst case bound will hold.
One well known example is copying a FIFO queue that is modeled with two stacks \cite{HAH2012}. There one starts with pushing the elements on the first stack and when the first POP operation is done,
one has to move all elements to the next stack to reverse their order. Then the next POPs are simple because most of the work has been done by moving the previous entries to the other stack.
Another example are the so called self-adjusting data structures \cite{doi:10.1137/0606031}.

Hofmann and Jost first applied amortized analysis by the potential method to first-order functional programs in \cite{HJ2003}. They annotated the types in the programs with the available resources
of the data structure and then introduced typing rules to reason about the resource consumption of functions. There they had the restriction that the potential was required
to be linear. The approach was later generalized to multivariate polynomial potential by Hoffmann \cite{HoffmannHofmann2010,Hoffmann2010,HAH2012,JH2011}. 
This was the starting point for many other investigations in this direction. Hofmann and Moser applied amortized analysis to term rewriting \cite{hofmann_et_al:LIPIcs:2015:5167,DBLP:journals/corr/HofmannM14},
Hoffmann refined his work, made it fully automatic and carried over the analysis to concurrent programs and programs in C and OCaml \cite{DBLP:conf/popl/HoffmannDW17,CHS2015,CHSR2014,Hoffmann2014,Hoffmann2015}.
Rodriguez introduced an amortized analysis for a fragment of Java, which features object oriented programming, polymorphic functions and monomorphic recursion
 \cite{HofRod:Esop13,HofRod09,DR2012,HR2012} and is called RAJA (Resource Aware JAva).
 
 Among other related work, mainly on resource analysis, are \cite{dantchevinfinite,Gradel-AutomaticStructures,HF,STVR:STVR1569,LPAR-21:Analyzing_Runtime_Complexity_via}, which use different methods than our approach.
 
The analysis system for RAJA by Rodriguez and Hofmann is the motivation for our tree constraints.
There the resource-type inference algorithm outputs conditions that must hold for the potential of the objects (represented as trees) in form of linear tree constraints.
One can determine the resource consumption of a RAJA program if one has a solution to its constraints. 
Recently, the list constraint satisfiability problem for RAJA was proven decidable \cite{LPAR-21:Decidable_linear_list_constraints}. 
In this paper, we generalize that argument to trees.
This decidability result enables us to analyze arbitrary RAJA programs with respect to their resource consumption.
Until now this was only possible for a subset of programs that need linear resources.

In the prototype implementation of RAJA\footnote{raja.tcs.ifi.lmu.de}, one
of the examples is sorting a list using merge-sort. There linear bounds are possible
by using static garbage collection, namely $\mathit{free}$ expressions that make additional
potential available for the further computation. There is research in this direction
\cite{DBLP:conf/iwmm/UnnikrishnanS09,AlbertGG10}, but by now there are still open questions about the realization
of a static garbage collector in Java. If we omit the $\mathit{free}$ expressions in the code,
the program is no longer analyzable, which means that it then requires nonlinear potential annotations. 
Nonlinear bounds make our analysis independent of
this construction and thus closer to real Java.

In addition to that, especially for bigger 
programs (like bank account models) or programs with auxiliary functions (like the sieve of Eratosthenes) or nested data structures, 
the constraint generation is very involved and often leads to superlinear potential. The same is true for cascades of recursive calls.

This paper is organized as follows: in the next section \ref{s2}, we formulate the tree constraint satisfiability problem and recall existing results.
Then we prove decidability of the tree case in section \ref{dec} and conclude in section \ref{s6}.

\section{Syntax and Semantics}\label{s2}
A tree constraint system is a set of pointwise linear inequalities between tree variables, 
as for example $x \geq r(x)+l(x)+y$, where $x,y$ are binary trees with labels $l$ and $r$ such that $r(x)$ is the right subtree of $x$ (and
$l(x)$ the left subtree.) In our setting, these tree variables can be instantiated with infinite trees that contain a nonnegative rational number in each node.
The degree of the trees is arbitrary but finite.

Let $L$ be a finite set of tree labels, $\lambda$ a variable, $n$ a number and let $\lozenge(x)$ denote the root of tree $x$.
The formal syntax for the linear tree constraints is shown in Figure \ref{syntaxgeneral}.

\begin{figure}[h]
\caption{Linear Tree Constraint Syntax}
\label{syntaxgeneral}
 \begin{align*}
 t &\Coloneqq x | l(t), \text{ where $l \in L$ with $|L| <\infty $ } &(\text{Atomic tree})\\
 te &\Coloneqq t | te+te  &(\text{Tree term})\\
 c &\Coloneqq te \geq te&(\text{Tree constraint})
\end{align*}
\end{figure}
In addition to the tree constraints, we have arithmetic constraints given for the numbers in selected nodes of the trees 
that take the form of an arbitrary linear program with integer coefficients.
They are the same as tree constraints with the difference that they can include numbers and hold only for the roots, which are arithmetic variables. 
An example for an arithmetic list constraint (where the root symbol $\lozenge$ becomes $head$ and the only label is $tail$) 
is $head(x)+ head(y)  \geq 2+ head(tail(z))$.
Figure \ref{syntaxConstrA} gives the syntax for them.
\begin{figure}[h]
\caption{Arithmetic Constraint Syntax}
\label{syntaxConstrA}
 \begin{align*}
 v &\Coloneqq n  | \lambda |  \lozenge(t) &(\text{Atomic arithmetic expression})\\
 h &\Coloneqq v | h+h  &(\text{Arithmetic term})\\
 c &\Coloneqq h \geq h&(\text{Arithmetic constraint})
\end{align*}
\end{figure}

Each inequality over tree variables corresponds to infinitely many inequalities over arithmetic variables (i.e. variables for the numbers in the nodes.)
Thus the problem to decide whether a set of tree constraints is simultaneously satisfiable can not directly be reduced to feasibility of a (finite) linear program.

A solution of the tree constraints is a set of infinite trees for which the constraints hold pointwise for each number in the nodes.
More precisely, a constraint $x \geq y$ holds for concrete trees $t_1,t_2$, if $\lozenge(t_1)\geq \lozenge(t_2)$ and for all labels $l$ 
(denoting the immediate subtrees) holds $l(t_1) \geq l(t_2)$ 
(cf. rule (Label) in Figure \ref{rules}).
\begin{figure}[h]
\caption{Tree with infinite number of different subtrees}\label{tree}
\begin{center}
\begin{tikzpicture}[scale=1]
\Tree
[.1
[.3
 [.9 
  [.27
  [.$\cdots$ ][.$\cdots$ ]
  ]
  [.18
  [.$\cdots$ ][.$\cdots$ ]
  ]
 ] 
 [.6 
 [.18
  [.$\cdots$ ][.$\cdots$ ]
  ]
  [.12
  [.$\cdots$ ][.$\cdots$ ]
  ]
 ]
] 
[.2
 [.6 
  [.18
  [.$\cdots$ ][.$\cdots$ ]
  ]
  [.12
  [.$\cdots$ ][.$\cdots$ ]
  ]
 ] 
 [.4 
 [.12
  [.$\cdots$ ][.$\cdots$ ]
  ]
  [.8
  [.$\cdots$ ][.$\cdots$ ]
  ]
 ]
] ].1
\end{tikzpicture}
\end{center}
\end{figure}
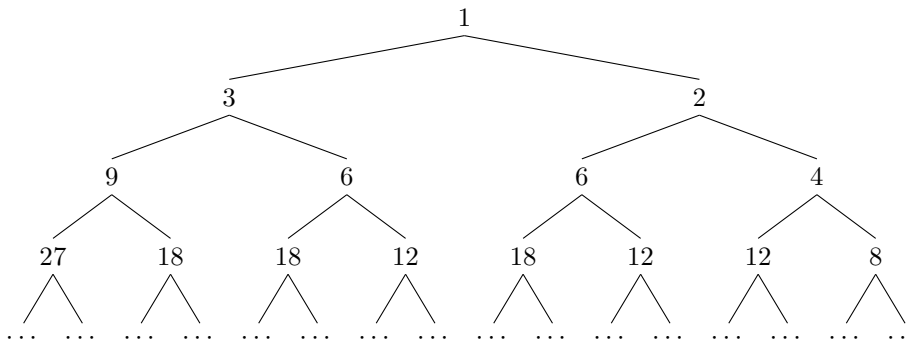
\begin{example}\label{infin}
 The system consisting of the arithmetic constraint 
  $\lozenge(t_1) = 1$
and the tree constraints \begin{align*}
r(t_1) \geq  2t_1, 
l(t_1) \geq 3t_1,
t_1 \geq lr(t_1),
\end{align*}
where the $\lozenge(\cdot)$ symbol denotes the variable in the root of a tree,
is unsatisfiable, because it implies 
$ 1  = \lozenge(t_1) \geq \lozenge(lr(t_1)) \geq 2\lozenge(l(t_1))\geq 6\lozenge(t_1)=6.$
The system 
$\lozenge(t_2) = 1, r(t_2)= 2t_2,
 l(t_2)= 3t_2$
has the solution in Figure \ref{tree}. 
The subtrees can be computed by
duplicating the value in the root of the subtree when going right and  
multiplying by three when going left:
\begin{align*}
 \forall w \in (l|r)^*: \lozenge(w(t_2)) = 2^i3^j, \text{ $i$ = number of $r$'s in $w$, $j$= number of $l$'s in $w$.}
\end{align*}
\end{example}

The following problems are closely related \cite{LPAR-21:Decidable_linear_list_constraints}.
\begin{mdframed}
\begin{itemize}
  \item \textbf{Skolem-Mahler-Lech Problem (SML)}

Given: A homogeneous linear recurrent sequence of degree $k$ with initial 
values 
$b_1,\dots,b_k$ and constant rational coefficients $a_1,\dots,a_k$ of the form 
\begin{align*}
 &x_n = a_1x_{n-1}+\dots+a_{k}x_{n-k}, n> k\\
 &x_1=b_1,\dots,x_{k} = b_{k}\\
 &a_i, b_i \in \mathbb{Q},b_i\geq 0 \text{ for all } i= 1,\dots,k, a_k \neq 0,
\end{align*}

Asked: Is there an index $n$ such that $x_n = 0$?

\item \textbf{List Constraint Satisfiability Problem (LC)}

Given: A finite system of list constraints (constraints over trees of degree 1 with label $tail$) and arithmetic constraints

Asked: Is there a set of lists, which simultaneously satisfies all constraints 
in of the system in $\mathbb{D} = \mathbb{Q}_0^+\cup \{\infty\}$?

\item \textbf{Tree Constraint Satisfiability Problem (TC)}

Given: A finite system of tree constraints and arithmetic constraints

Asked: Is there a set of trees, which simultaneously satisfies all constraints 
of the system in $\mathbb{D}$?
 \end{itemize}\end{mdframed}
 In \cite{LPAR-21:Decidable_linear_list_constraints}, it is shown that SML can be reduced to LC. 
 Thus LC and TC are very hard and probably undecidable problems; 
 at least the decidability status of the famous and NP-hard SML problem is still unknown.
 This led to the consideration of
 \emph{unilateral} constraints, that can be shown to be sufficient for our purposes and that are considerably easier to solve.
 
 \begin{definition}
  A unilateral tree constraint is a constraint with only one summand on the greater side of the inequality (i.e. of the form $t \geq te$ according to Figure \ref{syntaxgeneral}, 
  or equivalently $t \geq t_1 +\dots+ t_n$).
  We call unilateral tree constraints UTC.
 \end{definition}
 For instance, the satisfiable tree constraints in Example \ref{infin} are not unilateral.
 
It follows from the nonnegativity of the coefficients on the right hand side, that UTC is a proper fragment of TC.
 Indeed, it was shown that for the list case there exists a polynomial decision procedure by reduction to linear programming 
 \cite{LPAR-21:Decidable_linear_list_constraints}.
 This paper shows that UTC is also decidable. 
 In contrast to ULC (Unilateral List Constraints), our decision procedure is not polynomial in the size of the input.
 
 We also remark that according to the recurrence-like syntax nearly all constraints have only nonlinear solutions.
 For instance, in the case of lists, the linear system by Hofmann and Rodriguez can only (partially) treat periodic lists.
 In the tree case, we have analogous growth rates as for lists in \cite{LPAR-21:Decidable_linear_list_constraints}. 
 This means, as soon as we have a tree constraint with sums like e.g. $lrlr x\geq lrx+lrx$, the tree exhibits exponential growth.

\section{Decidability}\label{dec}
In this section we establish our main theorem, namely that satisfiability of unilateral linear tree constraints is decidable.
The proof is structured as follows: 
We observe that unsatisfiability is semi-decidable.
We show that we can reduce a set of constraint systems that contains all satisfiable ones to linear programming
using the following arguments:
\begin{itemize}
  \item We describe how to derive inequalities following from a set of constraints using a sound and complete proof system,
  \item characterize the set of trees greater than a fixed tree as a regular language,
  \item use these languages to find all trees bounded from above and from below,
  \item show that all other trees can be set to zero or infinity without changing the satisfiability properties of the system, and finally
  \item reduce the constraints to an equisatisfiable linear program.
 \end{itemize}
This means satisfiability is semi-decidable. Both together imply that satisfiability is decidable.

\subsection{Unsatisfiability is Semi-Decidable}
From now on we are in the realm of UTC and omit the word "unilateral".
 An \emph{unfolding step} for a constraint $x \geq S_1+\dots+S_n$, where $S_i$ are (sums of) tree variables,
 consists of adding the arithmetic 
 constraint $\lozenge(x) \geq \lozenge(S_1)+\dots+\lozenge(S_n)$
 and application of the (LabelSum) rule in Figure \ref{inf} to obtain the constraints for the next step.
 
 \begin{figure}[h]
\begin{mdframed}
\caption{Label application rules}
\label{inf}
$$\infer[\mathrm{(LabelSum)}]{TC  \vdash l(x) \geq \sum_i l(S_i)}{S= \sum_i S_i~~~~~~~~~~ TC  \vdash x\geq S}$$\\
$$\infer[\mathrm{(Root)}]{TC  \vdash \lozenge(x) \geq \sum_i \lozenge(S_i)}{S= \sum_i S_i~~~~~~~~~~ TC  \vdash x\geq S}$$
 \end{mdframed}
\end{figure}
 
 Each such step delivers a new, bigger set of arithmetic constraints that can be seen as a linear program.
 We have a succession of programs $(P_i)_{i\geq 0}$.
\begin{lemma}\label{contr}
 The constraint system $(AC,TC)$, where $AC$ is a set of arithmetic and $TC$ a setof tree constraints, 
 is unsatisfiable if and only if one of the linear programs $P_i$ is unsatisfiable.
 \end{lemma}
\begin{proof}
 The proof is basically the same as the compactness proof for infinite dimensional 0-1-programming in \cite{Comp}.
\end{proof}
Thus, if there is a contradiction, we find it, but if the system is satisfiable, this will not terminate.
In the remainder of this section, we give a procedure that terminates in the satisfiable case.

\subsection{The Set of Trees Greater than a Fixed Tree is a Regular Language}

We now describe the implications of a constraint system as given in Figure \ref{rules}.
Intuitively, if a constraint $x\geq y$ holds for trees $x$ and $y$, then also each subtree of $x$ is greater than or equal to the subtree of $y$ with the same label, 
and similarly the root of $x$ must be greater or equal to the root of $y$. Further, the greater-or-equal relation must be transitive.

\begin{figure}[H]\caption{Proof system for unilateral tree constraints}
 \label{rules}
$$\infer[\mathrm{(Reflexivity)}]
{TC \vdash ux \geq ux}{}$$ \\
$$\infer[\mathrm{(Label)}]{TC  \vdash l(x) \geq l(y)}{TC  \vdash x\geq y}$$\\
 $$\infer[\mathrm{(Transitivity)}]
 {TC \vdash ux \geq z}{x \geq y_1+...+y_n \in TC ~~~~~~~ TC \vdash uy_i \geq z}$$
\end{figure}

Tree expressions are of the form $ux$ where $u:\Sigma^*$ and $\Sigma$ is the set of tree labels like $\mathit{left,right}$, etc., and $x$ is a variable. 
We use letters $x,y,z$ for variables and for tree expressions.

The judgment $TC \models x \geq y$ , where $x$ and $y$ are expressions, has the
meaning that $x\geq y$ follows semantically from the
tree constraints in $TC$. That is, every valuation that satisfies $TC$ also
satisfies $x\geq y$.
The judgement $TC \vdash x \geq y$ means that the inequality $x \geq y$ is derivable by the rules in Figure \ref{rules}.
\begin{theorem}
 The proof system in Figure \ref{rules} is sound and complete (i.e. $TC \models x \geq y\Leftrightarrow TC \vdash x \geq y$).
\end{theorem}

\begin{proof} 
Soundness is trivial. 
For completeness we argue as follows.
Let $T$ be the set of all tree expressions over the variables in $TC$ and
define a graph $G=(V,E)$ where $V=T$ and 
\begin{align*}
E = &\{(x,y) | \text{ there is a
constraint } x' \geq y_1+\dots+y_n \in TC \text{ and }\\ &u:\Sigma^*.x=ux'\text{ and }
y=uy_i \text{ for some }i\}.
\end{align*}
Now fix a tree expression $x_0$ and define a valuation
$\eta$ in such a way that $\eta(\lozenge(z)) = 0$ if $z$ is reachable from $x_0$ in $G$
and $\eta(\lozenge(z))=\infty$ otherwise. We claim that $\eta$ satisfies $TC$. Indeed,
suppose that $x \geq y_1 +\dots+ y_n$ is a constraint in $TC$. We must show that
$\eta(\lozenge(ux)) \geq \eta(\lozenge(uy_1))+\dots+\eta(\lozenge(uy_n))$ holds for all $u:\Sigma^*$. 
Now if $ux$ is unreachable from $x_0$ then $\eta(\lozenge(ux))=\infty$ and the inequality
holds. On the other hand, if $ux$ is reachable then $uy_1,\dots,uy_n$ are also
reachable and the inequality holds as well. 

Now suppose that $x_0\geq y$ is an inequality that is not derivable from $TC$. In this case, $y$ is not
reachable from $x_0$ in $G$. The valuation $\eta$ constructed above then
satisfies $TC$ yet $\eta(x_0)=0$ and $\eta(y)=\infty$ so $x_0\geq y$ is not a semantic
consequence of $TC$.
\end{proof}
As a next step, we are interested in the set $L_z^{\geq}$ of tree expressions greater or equal to a fixed tree expression $z$; in short all $x$ such that $TC \vdash  x\geq z$.
Let us define the language $L_{x,y}\coloneqq \{u\mid TC\vdash ux \geq y\}$ with $x$ and $y$ fixed tree variables 
as an auxiliary step to compute $L_{z}^{\geq}$.

\begin{theorem}\label{regular}
 The language $L_{x,y}$ is regular.
\end{theorem}
\begin{proof}
We construct a finite automaton that accepts a word $w$, if and only if $TC \vdash wx\geq y$.

With the proof system in Figure \ref{rules}, we can first build a stack automaton $\mathcal{A}$ from $TC$, 
that reads no input and such that $u:L_{x,y}$ if and only if
 $\mathcal{A}$ accepts beginning from stack $ux$.
 
We give the idea for the construction of a slightly more general stack automaton, namely 
a stack automaton that accepts a word $vxyw^r$ if and only if the constraints imply $vx\geq wy$.
Acceptance is by empty stack, and we start by writing $v$ on the stack  while we are in a so-called "write-state",
 then go into a  state named $"x"$, there modify it nondeterministically and without reading from the input, as the constraints describe 
 (possibly going to state $"y"$ for another variable $y$).
 After that, we leave state $"y"$ (or $"x"$) and go into a "compare-state" where we
 compare the obtained stack with $w$ and empty it if they both are equal.
\begin{example}
Consider the constraints
 \begin{align*}
  lr(x) &\geq  rr(x),\\
    lr(x)&\geq l(y).
 \end{align*}
The stack automaton $M = (Z, \Sigma, \Gamma, \delta, z_0, \#)$
such that 
\begin{align*}
 Z = \{z_0,z_{\infty},z_x,z_y,z',z''\},\Sigma = \{ l,r,x,y\},\Gamma = \{L,R,\#\}
\end{align*}
and the transition relation $\delta$ is defined as depicted in Figure \ref{automat}. Note that the lower case 
input symbols in $\Sigma$ correspond to the according upper case letters in the stack alphabet $\Gamma$.
Here $z_0$ is the write-state, $z_{\infty}$ the compare-state and $z',z''$ are auxiliary states. 
For instance, in Figure \ref{automat}, the auxiliary states are used as intermediate steps to rewrite $lr(x)$ to $rr(x)$ or to $l(y)$.
We use the usual notation 
with triples for the current state, the read input symbol and the stack content, that are then mapped to the next state and the new stack content by $\delta$. 
In the picture, the triples on the arrows mean the input symbol, the stack before and the stack after the transition.
The symbol $a \in \Sigma$
\begin{figure}
\caption{Example stack automaton}
\begin{minipage}[right]{.5\textwidth}
 \label{automat}
 \begin{align*}
\delta(z_0,a,B) &= (z_0,AB),\\
             \delta(z_0,x,B) &= (z_x,B),\\
                 \delta(z_0,y,B) &= (z_y,B),\\
                 \delta(z_x, \epsilon, R) &=(z',\epsilon),\\
		\delta(z' \epsilon, L) &= (z'',\epsilon),\\
		\delta(z'', \epsilon, B)&= \{(z_x,RRB),(z_y,LB)\},\\
		\delta(z_x, x, B) &= (z_{\infty},B),\\
		\delta(z_y, y, B) &= (z_{\infty},B),\\
		\delta(z_{\infty}, r, R) &= (z_{\infty},\epsilon),\\
		\delta(z_{\infty}, l, L) &= (z_{\infty},\epsilon),\\
		\delta(z_{\infty}, \epsilon ,\#) &= (z_{\infty},\epsilon),\\
		B \in\Gamma^*,a &\in\{l,r\}.
\end{align*}
\end{minipage}
\begin{minipage}[left]{.5\textwidth}
\includegraphics[scale=.3]{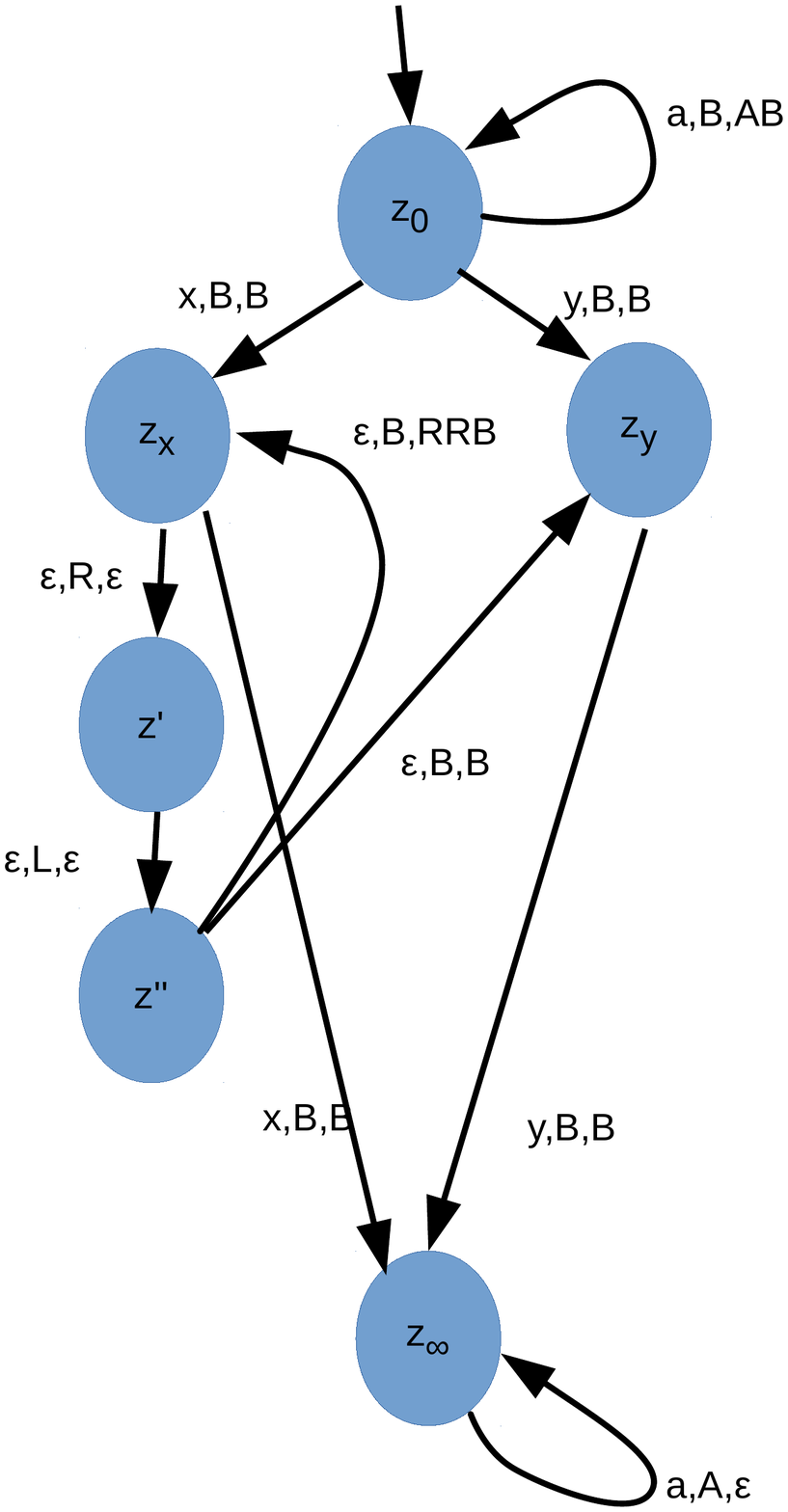}
\end{minipage}
\end{figure}
\end{example}
The language $L$ corresponding to this generalization is not always a regular language:
this can be seen with the Pumping Lemma. Assume $L$ is regular and let $p$ be the Pumping Lemma number and let the constraints be $lrx\geq lrlrx$ and consider the word $(lr)^ixx(rl)^{i+1}$, which is 
implied by the constraints and thus in $L$.
Then the word $\alpha = (lr)^p$ is obviously longer than $p$ and for each division of $\alpha= uvw$ in three words $u,v,w$,
for  instance with $v = lr$, then holds e.g. $\forall k. (lr)^{i+k}xx(rl)^{i+1} \in L$, 
which is not implied by the constraints. In any case we have that the label word before the $x$'s is longer than the label word after them, 
which is not a consequence of the constraint. This is a contradiction.

But we can use $L$, however, to show that the language $L_{x,y}$ with the second variable $y$ and its label word and $x$ fixed is always 
regular\footnote{In the list case, where we have only one label, this is a direct consequence of 
Parikh's theorem.}. For that, we build the above mentioned stack automaton $\mathcal{A}$ similar to the construction above. 
We can assume w.l.o.g. (possibly by introducing new states), that $\mathcal{A}$ has only transitions of the form 
$x\rightarrow y, x\stackrel{pop(a)}{\longrightarrow}y$, or $x\stackrel{push(a)}{\longrightarrow}y$, with $x,y$ states that belong to variables. 
We then define the set \begin{align*}
Q=\{(x,y) | x\rightarrow^*y\}
\end{align*}
and enumerate it using dynamic programming and the rules in Figure \ref{inf1}. 
Now we obtain for $L_{x,y}$, with $x,y$ both variables, the representation in Figure \ref{inf2}.
\begin{figure}[h]
\caption{Rules for Q}\label{inf1}
$$\infer
{(x,y):Q}{x\rightarrow y}$$\\
 $$\infer
 {(x,w):Q}{x\stackrel{push(a)}{\longrightarrow}y ~~~~~~~~~~(y,z):Q~~~~~~~~~z\stackrel{pop(a)}{\longrightarrow}w}$$
\end{figure}
\begin{figure}
\caption{Rules for $L_{x,y}$}\label{inf2}
$$\infer{\epsilon:L_{x,y}}{(x,y):Q}$$\\
$$\infer{ua:L_{x,y}}{u:L_{x''',y} ~~~~~~~~~~(x'',x'''):Q~~~~~~~~~x'\stackrel{pop(a)}{\longrightarrow}x''~~~~~~~~~(x,x'):Q}$$
\end{figure}
From this we can read off a finite automaton $\mathcal{B}$ for $L_{x,y}$ directly: the states are the states of $\mathcal{A}$, for each pair
$(x,y) \in Q$ we introduce an $\epsilon$-move from state $x$ to state $y$, and the $\mathcal{A}$-transitions $x'\stackrel{pop(a)}{\longrightarrow}x''$ 
are the nontrivial moves that consume the letter $a$.
\end{proof}

We then also have $L_{ux,y} = \{w \mid wu : L_{x,y}\},L_{x,vy} $  and $L_{ux,vy}$ regular. 
The disjoint union over the sets $L_{x,vy}$ for all $x$ equals $L_z^{\geq}$, for the expression $z = vy$.

From now on we omit the brackets for trees with prefixed labels and write $lx$ instead of $l(x)$.

\begin{example}
 Let the constraints be
 \begin{align*}
  lx \geq x,
  x \geq rz,
  lrz \geq lly,
  ly \geq y.
 \end{align*}                                                                                                                                         
Then $L_{x,y}=L_{x,ly}= l^+$ and $L_{lx,y} = l^*$.
The language $L_{lrz}^{\leq} = \{lr_z,ll_y,l_y,\epsilon_y\}$, where the subscript $x$ at label word $w$ means that $w \in L_{x,lrz}$.
\end{example}

Tree constraints systems without arithmetic constraints are always trivially satisfiable by setting all tree entries to zero.
Analogously, all nodes that have bounds only in one direction (i.e. are only implied to be greater than a set 
of arithmetic variables $a_i$ or only less) can be set to zero or infinity.
The only interesting case appears when we have subtrees $x$ whose root $\lozenge(x)$ lies between two arithmetic variables $a$ and $b$.
The set $L_a^{\geq}\cap L_b^{\leq} \coloneqq \{x\mid a \leq \lozenge(x) \leq b\}$ can be computed using the languages
$L_x^{\geq}$ and $L_y^{\leq}$ for certain subtrees $x,y$.
These trees are defined as the subtree starting at the point where the arithmetic variables $a,b$ are located.
For instance, if $a = \lozenge(lrrz)$, then the subtree $x$ is $lrrz$.
Thus we can write $L_a^{\geq}\cap L_b^{\leq} = \{x\mid TC \vdash y_a \leq x \leq y_b \}$, 
where $y_a$ (resp. $y_b$) is the tree with root $\lozenge(y_a) = a$ (resp. $b$).  
\begin{example}\label{irregular}
 Consider the constraints
 \begin{align*}
  \lozenge(x) &= 1,
lrx \geq x, lx \geq x, mx \geq x,
x \geq rlx, x \geq mlx.
\end{align*}
The language of trees greater with root greater than the root of $x$, less than the root of $x$, and equal to it are: 
\begin{align*} 
 L_{\lozenge(x)}^{\geq} &=(lr\mid l \mid m)^*x,
 L_{\lozenge(x)}^{\leq} = (ml\mid rl)^*x,\\ L_{\lozenge(x)}^{=} &=L_{\lozenge(x)}^{\leq}\cap L_{\lozenge(x)}^{\geq} = m(lr)^*l x= ml(rl)^*x. 
\end{align*}
This is obtained by iteratively applying the constraints and transitivity.
\end{example}

\subsection{Normal Form for Tree Constraints}
We now bring the constraints into a normal form to start our procedure. 
Constraints in this normal form all have a variable with label word of length $n$ on the left hand side, and all label words on the right are
at most of length $n$. The variables with label word of length exactly $n$ can be represented as a directed acyclic graph with an edge between $x$ and $y$ if and only if
$TC \vdash x \leq y$.
Further, there are no arithmetic constraints below level $n-1$ (i.e. for trees with label word of length more than $n-1$).

To obtain constraints in this normal form, we examine the form of each constraint. 
If there is only one label word of maximal length, we take this word and isolate it on the left side of the inequality.
If there is more than one such word, we write $k$-times repeated addition of the same summand $s$ as $k\cdot s$.
If then there is only one summand $s$ of maximal length, we bring it on the left hand side and divide both sides by $k$.
Otherwise, we  build $l$ constraints by bringing the $l$'th of the longest summands on  the left (possibly again by dividing by a positive integer).
Then, we apply the rules in Figure \ref{inf} to all thus obtained constraints until all label words on the left have the same length. 
The result is an equivalent constraint system (i.e. a system with exactly the same 
solutions) consisting of unfolded tree constraints and a new, bigger set of arithmetic constraints.

Depending on the position of this longest label word and according to the unilateral syntax, we have --- after bringing the longest label word on the left side --- 
three kinds of constraints:
\begin{itemize}
 \item \emph{lower bounds} are of the form $wx \geq w_1y_1+\dots+w_iy_i$, with all $w_i$ shorter than $w$.
 \item \emph{upper bounds} have the shape $wx \leq w'y- w_1y_1-\dots-w_iy_i$, with all $w_i$ and $w'$ shorter than $w$.
 \item \emph{undirected constraints} have two label words of the same length on both sides, as for instance $x \geq y+z, x \leq y-z$.
\end{itemize}

The last set can be transformed into a directed acyclic graph by removing cycles as follows:
If we can derive $x\geq y_1\geq y_2\geq \dots \geq x+R$ by just using transitivity 
(not label application, which makes it immediate to decide), then we conclude that $x = y_i \forall i$ and that $R$ is identical 
to the tree consisting only of zeros.
In this graph, we have now encoded upper and lower bounds simultaneously.

\begin{example}\label{exampl}
 Let the constraints be 
  $x \geq y+z, z\geq t,y \geq t. $
 They correspond to the graph with an edge from $x$ to $y$ and to $z$ and from $y$ and $z$ to $t$, shown
 in Figure \ref{dag}.
 \begin{figure}[h]
 \caption{Graph for Example \ref{exampl}}
 \begin{center}\label{dag}\begin{tikzpicture}
                   \node (x) {x};
                   \node (y) [below of = x, left of = x]{y};
                   \node (z) [below of = x, right of = x]{z};
                   \node (t) [below of = x, left of = z]{t};
                   \path[->] (x) edge (y);\path[->] (x) edge(z);\path[->] (y)edge (t);\path[->] (z) edge(t);
                   \end{tikzpicture}\end{center}\end{figure}
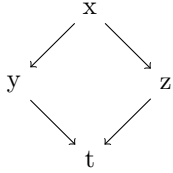
Then we add the four constraints $t \leq z,t\leq y, y \leq x-z, z \leq x-y$
to our system. We must traverse this graph in two directions (i.e. we need both kinds of bounds that are implied by it) 
to obtain an order in which we treat the nodes.
Why we need this, will become clearer in Example \ref{why}.
\end{example}

\subsection{Idea and Examples}

Before moving to Example \ref{why}, we will briefly explain the idea behind our procedure.
Afterwards, we will cover the technical details.
The intuition is that we label all nodes in the trees which are in a set $L_a^{\geq}\cap L_b^{\leq}$ for arithmetic variables $a$ and $b$ with sets of intervals,
in which the number in the node has to lie. 
These intervals are derived from the constraints.
Nodes with the same set of intervals are defined to  be in the same \emph{class}.
Then we show that in a subset of all constraint systems that contains the satisfiable ones, there are only finitely many different classes.
Last, we translate the statement that all these intervals are nonempty into a finite set of linear inequalities between the arithmetic variables.
This linear program is equisatisfiable to the constraints (i.e. if the intervals are nonempty, then there exists a solution with the valuation of each
$\lozenge(x)$ in the interval assigned to $\lozenge(x)$). 
If they are satisfiable, 
we thus get the answer in terms of a satisfiable linear program, and, in addition to that,
an assignment of a class to each of the nodes that can be seen as a certificate for satisfiability. 
Combining this with Lemma \ref{contr}, we have a decision procedure that either returns an unsatisfiable linear program implied by the constraints or a schematic notation
for the intervals that contain their solutions from which it is directly possible to compute a solution. 
\begin{example}\label{why}
Consider the (list-) constraints
\begin{align*}
\lozenge(x)=\lozenge(y) =1,
 x \geq y,
 lx \leq x,
 ly \geq y.
\end{align*}
They are equivalent to the system $\lozenge(x)= \lozenge(y)=1$ in conjunction with $C= \{c_1, c_2,c_3,c_4\}$, where
\begin{align*}
 c_1 \coloneqq ly \leq lx,
  c_2 \coloneqq lx \geq ly,
  c_3 \coloneqq lx \leq x,
 c_4 \coloneqq ly \geq y,
 \end{align*}
which is in normal form (with the constraint $c_1$ redundant in this case).
The constraints $c_3$ (resp. $c_4$ deliver the intervals $[0,1]$ for $\lozenge(lx)$ (resp. $[1,\infty ]$ for $\lozenge(ly)$). Then $c_1$ gives
us the interval $[1,1]$ for $\lozenge(ly)$ and $c_2$ gives the same interval for $\lozenge(lx)$. 
By duplicating the constraint, we ensure that we treat the nodes in subsequent 
levels of the DAG (as constructed above) correctly. Imagine we had only $c_1$ without $c_2$, then we would miss the bound on $\lozenge(ly)$.
In the next steps, we derive no new bounds any more.
\end{example}
The system in Example \ref{why} is satisfiable by the trees consisting only of 1s.
If we modify it slightly, it becomes unsatisfiable:
\begin{example}\label{why2}
 The constraints
$x \geq y, \lozenge(y) \geq 1,
 lx+lx \leq x,
 ly \geq y+y$
are equivalent to $\lozenge(x)\geq \lozenge(y)\geq 1$ and $D= \{d_1, d_2,d_3,d_4\}$, where
\begin{align*}
 d_1 \coloneqq ly \leq lx,
 d_2 \coloneqq lx \geq ly,
 d_3 \coloneqq lx \leq 0.5 x,
 d_4 \coloneqq ly \geq 2y.
\end{align*}
We have the intervals $[0, 0.5\lozenge(x)]$ and $[2\lozenge(y),\infty]$ for $\lozenge(lx)$ by $d_3$ and $d_2$, and $[2\lozenge(y),\infty]$ and 
$[0,0.5\lozenge(x)]$ for $\lozenge(ly)$ by $d_4$ and $d_1$. In the next steps, the factors will be 0.25 and 4, etc.
The list $x$ is exponentially decreasing, whereas $y$ grows exponentially. So if the roots of $x$ and $y$ are neither zero nor infinity, 
then no matter which number they are, $x$ will at some point be less than $y$.
Here we see that there is a contradiction, but we can not say after how many iterations we will find it. There the other part of our algorithm, namely Lemma 
\ref{contr}, applies. 
\end{example}
\subsection{Satisfiability is Decidable}
Before we prove the central fact (Theorem \ref{thm}) of this paper, we need some word-combinatorial preliminaries.
We say that two label words are \emph{dependent} if one is a suffix of the other.
The next two lemmas are well known and can be found for instance in \cite{CK1997}. 
\begin{lemma}\label{words}
 Let $u,s,t \in \Sigma^+$ such that $tu=st$. 
 There then exist $q,r\in\Sigma^*$ and $i\in\mathbb{N}$ such that $s=qr$, $u=rq$, $t=q(rq)^i$.
\end{lemma}
\begin{lemma} \label{words3}
 If for word $x,y,z$ holds
$x^ny^m=z^k$, with
$n, m, k \geq 2,$
then exists $t$ such that $x,y,z  \in t^*$.
\end{lemma}
 
\begin{lemma}\label{words2}
 Let $c$ be a unilateral tree constraint.
 If it is of the form
 \begin{align}\label{summe}
      a_1a_2\dots a_n x \geq c_1 \cdot a_2\dots a_n x +\dots+c_n \cdot a_n x+ c_{n+1}\cdot x
     \end{align}
with all $a_i \neq \epsilon,c_j \in\mathbb{N}_0$ and label words and $x$ a tree,
then it can be transformed (by application of labels from the left) into a constraint with all $a_k \in p^+$ for a suitable word $p$ and all other summands independent.
\end{lemma}

\begin{proof}We assume that there is a label word $t\in \Sigma^*$ which we can apply from the left such that all summands stay dependent of
$ a_1a_2\dots a_n x$. 
(If such a $t$ does not exist, then all summands are already independent.) This is,
 \begin{align}\label{gl}
  ta_1a_2\dots a_n\text{ has the suffixes } ta_2\dots a_n, ta_3\dots a_n,\dots ta_n,t.
 \end{align}
We apply Lemma \ref{words} to $ta_1a_2\dots a_n = st$  with $ u=a_1\dots a_n$ and obtain $r,q \in \Sigma^*$ such that $a_1\dots a_n = rq, t = q(rq)^i = q(a_1\dots a_n)^i$.

Thus $ta_2\dots a_n = q(a_1\dots a_n)^ia_2\dots a_n$.  According to (\ref{gl}), $ta_2\dots a_n$ is a suffix of $ta_1a_2\dots a_n= q(a_1\dots a_n)^{i+1}$.
This means that $a_1$ and $a_2\dots a_n$ commute. Thus there is $p_1$ such that both are  in $p_1^+$.

Similarly, $ta_1a_2\dots a_n$ has the suffix $ta_3\dots a_n=q(a_1\dots a_n)^ia_3\dots a_n$, and thus $a_1a_2$ and $a_3\dots a_n$ commute (see Figure \ref{fig} ,where words of the same length are written in boxes).
We can thus conclude that there is $p_2$ with $a_1a_2$ and $ a_3\dots a_n $ are in $ p_2^+$.
\begin{figure}
\begin{center}\begin{tabular}{l|l|l|} \hline
 &$a_1 a_2 \dots a_n$& $a_3 \dots a_n$\\\hline
$a_1 a_2$ &$a_3 \dots a_n a_1 a_2$ &$a_3 \dots a_n$\\
\hline
\end{tabular}\end{center}
\caption{Commuting words $a_1a_2$ and $a_3\dots a_n$}
 \label{fig}
\end{figure}
We proceed the same way until we obtain in the last step that $a_1\dots a_{n-1}$ and $a_n$ commute.
We now write $a_1 = p_1^{i_1}, a_2\dots a_n = p_1^{j_1}$ and $a_1a_2= p_2^{i_1}, a_3\dots a_n = p_2^{j_2}$, etc.

Application of Lemma \ref{words3} 
allows us to conclude from
$(a_1\dots a_n)^2 = p_1^{i_1+j_1}p_2^{i_2+j_2} = p_3^{2(i_3+j_3)}$
that $p_1,p_2,p_3 \in p^*$ for a certain $p.$ Thus all $p_i$ are in $p^+$ and  for all $i$, we have $a_i \in p^+$.
\end{proof}

\begin{theorem}\label{thm}
Satisfiability of linear tree constraints is semi-decidable.
\end{theorem}
\begin{proof}
 We assume that all constraints are in the normal form described above.
 Then we introduce an arithmetic variable for each node above level $n$.
 Recall that there are no arithmetic constraints below level $n-1$ and all left hand sides if the constraints have label word of length exactly $n$.
 We now calculate the sets $L_a^{\geq} \cap L_b^{\leq}$ for all pairs of arithmetic variables $a,b$.
 W.l.o.g. we can further assume that all $a, b$ are  nonzero and not infinity: for instance, we could try all variants of the constraints with additional $A\ni a_i = 0$
 or $B \ni a_i = \infty$ for all pairs of disjoint subsets $A,B$ of the set of arithmetic variables and such that all other $a_i$ are neither zero nor infinity.
 If one of them is satisfiable, we return this as a result.
 
 Our procedure starts with step 1 at level $n$ and assigns a set of intervals to each node. 
 For the lower bounds, which have the form $wx \geq w_1y_1+\dots+w_my_m$, with all $w_i$ shorter than $w$,
 we add the interval $[\sum_i\lozenge(w_iy_i),\infty]$.
For the upper bounds, that have the shape $wx \leq w'y- w_1y_1-\dots-w_my_m$, with all $w_i$ and $w'$ shorter than $w$,
we add the interval $[0, \lozenge(w'y)-\sum_i\lozenge(w_iy_i)]$.
For the undirected constraints, we observe the following.
The membership of all nodes in $L_a^{\geq} \cap L_b^{\leq}$ ensures that we have already an interval for the starting nodes of the DAG constructed above.
We traverse it in both directions and add for constraints $x \geq y_1+\dots +y_m$ (resp. $x \leq y-z_1-\dots-z_m$) the new set of intervals 
$\{[\sum_i a_i,\infty] \mid [a_i,b_i] \text{ is an interval for $\lozenge(y_i)$}\}$ (resp. $\{[0, b-\sum_i c_i] \mid [a,b] \text{ is an interval for $\lozenge(y)$ and }[c_i,d_i]
\text{ is an interval for $\lozenge(z_i)$}\}$). 

Further, we set all nodes that have bounds in only one direction to $[0,0]$ or $[\infty, \infty]$. We
denote the set of intervals for node $\lozenge(x)$ with $I(x)$, and $I_n$ is the set of all $I(x)$ obtained until step $n$.

The unilateral constraint syntax allows us to define a meaningful addition and subtraction on interval sets that formalizes how we compute new interval sets.
\begin{align*}
 I(x) &+I(y) = \{[a+c,\infty]\mid [a,b] \in I(x),[c,d]\in I(y)\},\\
 I(x) &-I(y) = \{[0,b-c]\mid [a,b] \in I(x),[c,d]\in I(y)\}.
\end{align*}
\begin{observation}
The order of evaluation does not play any role for sums of interval sets
(i.e. $I(x) -I(y) -I(z)= I(x) - (I(y) +I(z))$).
\end{observation}
To prove this, let w.l.o.g. be $I(x)= [a,b], I(y)= [c,d],I(z) = [e,f]$.
Then \begin{align*}
      [a,b]-[c,d]-[e,f]&=[0,b-c-e]=[0,b-(c+e)] \\&= [a,b]-[c+e,\infty] =[a,b]-([c,d]+[e,f]).
     \end{align*}
In step $n+1$, we apply the rule (LabelSum) in Figure \ref{inf} to the constraints to make their left sides have a label word of length increased by 1.
Then, for the lower bounds we no longer necessarily have arithmetic variables as roots of the trees on the right, but also nodes equipped with intervals.
Thus, we proceed in a similar way as for the undirected constraints in level $n$, namely add the intervals that can be derived from the variables on the right.
We do the same for the upper bounds and the undirected constraints.
More precisely,
for the lower bounds, we set $I(x) = I(y_1)+\dots +I(y_m)$, and for the upper bounds $I(x) = I(y)- I(z_1)\dots -I(z_m)$.

We claim that after a finite amount of steps, no new intervals are derived any more.
This is, if we see the set of intervals that belong to a node as its \emph{class}, then there are only finitely many different classes.
The reason is that if the intersection of one of the interval sets would be constantly shrinking, we would infinitely often add a nonzero number to the lower bound 
or subtract a nonzero number from the upper bound or divide the upper bound by a positive integer 
(by the assumption that all arithmetic variables are neither zero nor infinity).
But since all considered nodes are bounded from above and below, we would at some point obtain a contradiction (see Example \ref{why2}).
Thus it is enough to give a criterion ensuring that we need no longer search for new classes because we found all.
Having this, the condition that the intersection of all intervals that belong to the same node is nonempty delivers an equisatisfiable linear 
program.
We now define $S$ as the least common multiple of all differences of label word lengths that appear in the constraint system.
For instance, $S$ for the single constraint $lrrlx \geq x+lx$ is $12 = \mathrm{lcm}(4,3)$. 
Note that in the list case, $S$ is a bound on the period length of the solution lists
(cf. \cite{LPAR-21:Decidable_linear_list_constraints}). 

The criterion looks as follows:
If in $S$ iterations no new interval sets for the nodes in $L= \cup_{a,b}(L_a^{\geq} \cap L_b^{\leq})$ are derived any more 
(i.e. for each node $x$ on a certain level and word $p$ with $|p|=S$, the intersection of all intervals that belong to $px$
is equal to the intersection of the intervals for $x$), 
then we have found all of them.

There are two things to show, namely 
 that the premise of this criterion implies $I = I_n$ for a $n\in \mathbb{N}$ and
 that this premise will finally hold.
 
\begin{claim}[Part 1]
\label{erstens} 
If there is a $n\in \mathbb{N}$ such that for all $x\in L$ on level $n,\dots, n+S$ and for all label words $p$ with $|p|= S$, the set $I(px)$ is equal to
$I(x)$, then $I = \cup_{j\in \mathbb{N}} I_j = I_{n+S}$.
\end{claim}

\begin{claim}[Part 2]
\label{zweitens} 
There is a $n\in \mathbb{N}$ such that for all $x\in L$ on level $n,\dots,n+S$ and for all label words $p$ with $|p|= S$, the set $I(px)$ is equal to
$I(x)$.
\end{claim}
To prove the first, 
we show that for all $k= 0,\dots S$ and for all $l\in \mathbb{N}$, we have $I_{n+k+l\cdot S} = I_{n+k}$. 
We consider three cases.
If we have a lower bound constraint $px \geq \sum_i y_i$, we know that for all label words $p$ with 
$px \in L $
and $|p|= S$, this implies $qpx \geq \sum_i qy_i$. The lower bounds of the intervals
for $qy_i$ are not stronger than those for $y_i$. This follows from the assumption if $qy_i \in L$. 
It is also true if $qy_i\notin L$, because then the interval for $qy_i$ must be $[0,0]$ (since it is less or equal to $qpx$, which is at most $b$ and 
so it can not be $[\infty,\infty]$)
and thus it delivers no new bounds at all.
We mark this property by $(\star)$.
Similarly, in case of upper bounds $px \leq y-\sum_i z_i$, 
the upper bounds for $qy$ and the lower bounds of the intervals for $qz_i$ are not stronger than those for $y$ and $z_i$. 
Again, if $qy,qz_i\in L$, this is a consequence of the assumption, and if $qy \in L$ and $qz_i\notin L$, then 
$qz_i$ has interval $[0,0]$. Last, if $qy \notin L$ then $qy$ has interval  $[\infty, \infty]$ and delivers no new bounds.
This property is called $(\star\star)$.
In these two cases, the right hand side in $(\star)$ and $(\star\star)$ is on level less or equal to the level on the left.

Thus we may assume that all $py_i,qy,qz_i \in L$. For all label words $q,p$ with $qpx\in L$ and $|q|,|p|=S$, we have $I(qy_i) = I(y_i),
I(qy)= I(y),I(qz_i)= I(z_i)$ and
the set of intervals for $qpx$, which is the sum  
of the intervals for the $qy$ (resp. the difference between the intervals of $qy$ and the $qz_i$) is (after intersection) 
not smaller than the set of intervals for 
$qpx$ (according to $(\star)$ and ($\star\star$)), and also not smaller than the interval for $x$.
More precisely, we have
\begin{align*}
 I(qpx) &= I(qy_i)+\dots+I(qy_m)=I(y_i)+\dots+I(y_m) = I(px)= I(x), \text{ or resp.}\\
  I(qpx) &= I(qy)-I(qz_i)-\dots-I(qz_m)= I(y)- I(z_1)-\dots-I(z_m) = I(px)= I(x).
\end{align*}
The last case is if we have an undirected constraint $x \geq y+z$. This implies $px \geq py+pz$,
and with $I(px)= I(x),I(py)= I(y),I(pz) = I(z)$, we have for the new interval obtained from the undirected constraint
$I_{new}(px)= I(py)+I(pz)= I(y)+I(z)= I(x)$ etc. If there are no changes in the intervals for $p,q$ of length $S$, then there are no changes at all and $I= I_{n+S}$.

To prove the second part of the claim, we assume that
for all levels there is a $p$ of length $S$ and $x$ on that level (optionally plus a number between 1 and $S$) 
such that $I(px)\neq I(x)$.
All constraints on the variables in $L$ (except a subset of the undirected constraints where all label words have the same length)
imply constraints of the form $qy\geq z+R$ (resp. $qy\leq z-R$) 
with $|q|=S$. For the lower bounds, we have to consider all possibilities for the choice of $y$,
whereas for the upper bounds there is only one positive summand.
We can therefore assume that $I(px)\subset I(x)$ because of the choice of $S$: else, if $pqy$ had strictly weaker lower bounds 
(resp. strictly weaker upper bounds) than $pz$, the interval $I(pqy)$ would contain $I(pz)$ and thus $qy$ could not have $z$ as a bound.
Then we have
$I(pqy) =  I(pz)+I(R)$ (resp. $I(pqy) =  I(pz)-I(R)$).
We now assume w.l.o.g. that $y$ plays the role of the $x$ above and that $z = y$ holds\footnote{If this is not the case for the initial $y$, the next candidate for $x$ is $z$.}.
So we have a constraint $px\geq x+R$ or $px \leq x-R$ derivable just by unfolding using the (LabelSum) rule in Figure \ref{inf}.
We just treat the first since both are similar.

According to Lemma \ref{words2}, either some label words of the summands in $R = r_1+\dots+r_m$ and 
$p$ are powers of the same path $t$, or in the next step all summands in $qR\coloneqq qr_1+\dots+qr_m$ are independent of $qpx$ for all $q$ of length $S$.
If the second happens, if no tree $t^ix$ is reachable\footnote{No lower bound constraint on any summand $t'x$ in $R$ with $t'\in t^*$ on the right exists --- which is decidable according to 
Theorem \ref{regular}.} 
from $R$, then either this constraint delivers no new bounds below level $n+S$, or 
$R$ must contain a tree $z$ that is at least constant when seen as a list $(p'^iz)_i$ along a path $p' \in t^*$ of length $S$. This implies $p= p'$, which again implies
$ppx= t^jpx = t^{2j} x \geq t^j x +t^{kj}z+ R'$, and that means $p^lx \geq p^{l-1}x+p^{l+k-2}z$ holds --- just like in the first case. 
Overall, we have that either the constraint does not deliver an infinite amount of new bounds or has the form of a strictly increasing list along the path 
$(\lozenge(p^i))_i$.  
In both cases, this is a contradiction, since only finitely many $p^ix$ can be in $L$, 
thus $I(p_ix)$ is at some point equal to $[\infty,\infty]$ and then stops changing.

Claim 1 and Claim 2 ensure that in a set of cases including the satisfiable ones, 
we will only derive finitely many different intervals. 
Thus the problem to decide whether the values  of the arithmetic variables can be chosen such that
these intervals are all nonempty can be solved by linear programming. If and only if they can be chosen this way, 
the constraints are satisfiable.

This completes our proof. \end{proof}
 Combined with the semi-decidable unsatisfiability, we can decide UTC. 
\begin{example}
 Let the constraints be 
 \begin{align*}
 \lozenge(y) = 1, y\geq ly,y\geq ry \text{   and   }
  lx \geq x+y,rx \geq x+y,x\geq lr x.
 \end{align*}
Then all nodes in $(l|r)^+y$ are assigned the intervals $[0,0]$.
Similarly, $l^+x$ and $r^+x$ and all other nodes are set to $[\infty,\infty]$, except those nodes that have bounds in two directions (i.e. are in $L$).
The only nodes in $L$ are the roots of $lr^*x$ and $\lozenge(y)$.
Thus we only need to compute intervals for $lr^*x$.
The root $\lozenge(x)$ gets the interval $[1,1].$ On level two, there is no node in $L$.
Then, on level three, $\lozenge(lrx)$ is labeled with the intervals $[1,\infty]$ and $[0,1]$.
Thus their intersection is equal to $[1,1]$. The same happens on level $4,6,8,$ etc.
Indeed, we can easily check that another solution than one with $\forall w \in (lr)^*.\lozenge(wx)=1$ is not possible.
\end{example}

\section{Conclusion}\label{s6}
 We have proven that linear constraints over infinite trees, as generated by an automatic resource type inference for the language RAJA,
 are decidable. Our approach uses finite automata to generalize the list constraint theory to trees.
 For the latter, satisfiability was previously proven decidable in polynomial time. 
 In contrast to that, our algorithm for trees needs exponential time, because the number of the linear programs that 
 we reduce the problem to is exponential in the size of the input.
 
 With this result, we can now analyze arbitrary object oriented (RAJA-) programs with respect to their resource consumption.
 We can read off upper bounds on the memory usage from the solutions of the constraints. 
 The next parts of our planned future work include a more concrete description of minimal constraint solutions as closed formulas and
 an implementation based on the existing RAJA tool.
 
 We also will further investigate possibilities to increase the efficiency by optimizing the decision procedure. 
 In order to develop a powerful tool for analyzing real programs, we plan to add Java features (as exceptions, garbage collection, concurrent threads) to RAJA.
 Another approach would be implementing an automatic translation from Java code to an equivalent RAJA program (w.r.t. the resource consumption),
 which can then be analyzed using our results.


\bibliography{LiteraturePaper}
\nocite{*}

\end{document}